\documentclass[a4paper,onecolumn,11pt,unpublished]{quantumarticle}
\pdfoutput=1
\usepackage[utf8]{inputenc}
\usepackage[english]{babel}
\usepackage[T1]{fontenc}
\usepackage{amsmath,amsthm}
\usepackage{hyperref}
\usepackage[numbers]{natbib}
\usepackage{appendix}
\usepackage{physics}
\usepackage{tikz,algorithm}
\usetikzlibrary{quantikz2}
\usepackage[noend]{algpseudocode}
\usepackage{thm-restate}
\usepackage{nicematrix}

\newif\ifarXiv
\arXivtrue

\newtheorem{theorem}{Theorem}
\newtheorem{lemma}{Lemma}
\newtheorem{definition}{Definition}
\newtheorem{corollary}{Corollary}

\newcommand{\CNOT}{\mathsf{CNOT}}
\newcommand{\QFT}{\mathsf{QFT}}
\newcommand{\MCX}{\mathsf{MCX}}

\newcommand{\Ig}{\mathsf{I}}
\newcommand{\Ug}{\mathsf{U}}
\newcommand{\Xg}{\mathsf{X}}
\newcommand{\fanout}{\textsc{Fan-out}}
\newcommand{\Fg}{\mathsf{F}}
\newcommand{\ladder}{\textsc{Ladder}}
\newcommand{\Lg}{\mathsf{L}}

\newcommand{\ie}{i.e.}
\newcommand{\etal}{et \; al.}
\newcommand*{\eqdef}{\stackrel{\text{def}}{=}}
\newcommand{\floor}[1]{\left\lfloor #1 \right\rfloor}
\newcommand{\ceil}[1]{\left\lceil #1 \right\rceil}

\newcommand{\Rbra}[1]{\left( #1 \right)}

\newcommand{\bigO}[1]{O\Rbra{#1}}

\newcommand{\bs}[1]{\boldsymbol #1}
\renewcommand{\mod}[1]{\ \mathrm{mod}\ #1}

\begin{document}

\title{Ancilla-free Quantum Adder with Sublinear Depth}
\author[1]{Maxime Remaud}
\orcid{0009-0008-1597-3661}
\author[1,2]{Vivien Vandaele}
\affil[1]{Eviden Quantum Lab, Les Clayes-sous-Bois, France}
\affil[2]{Universit\'e de Lorraine, CNRS, Inria, LORIA, F-54000 Nancy, France}
\orcid{0000-0002-9330-4999}
\maketitle

\begin{abstract} 
    We present the first quantum adder with sublinear depth and no ancilla qubits. Our construction is based on classical reversible logic only and employs low-depth implementations for the $\CNOT$ ladder operator and the Toffoli ladder operator, two key components to perform ripple-carry addition.
    
    Namely, we demonstrate that any ladder of $n$ $\CNOT$ gates can be replaced by a $\CNOT$-circuit with $\bigO{\log n}$ depth, while maintaining a linear number of gates. We then generalize this construction to Toffoli gates and demonstrate that any ladder of $n$ Toffoli gates can be substituted with a circuit with $\bigO{\log^2 n}$ depth while utilizing a linearithmic number of gates. This builds on the recent works of Nie $\etal$ \cite{NZS24} and Khattar and Gidney \cite{KG24} on the technique of \textit{conditionally clean ancillae}. By combining these two key elements, we present a novel approach to design quantum adders that can perform the addition of two $n$-bit numbers in depth $\bigO{\log^2 n}$ without the use of any ancilla and using classical reversible logic only (Toffoli, $\CNOT$ and $\Xg$ gates). We also present new constructions for incrementing and adding a constant to a quantum register.
\end{abstract}

\section{Introduction}

In the 1990s, the discovery by Shor of a polynomial-time quantum algorithm for factoring numbers \cite{Sho97} was a catalyst for a surge of interest in quantum computing and its potential to accelerate the resolution of various problems. At the core of this algorithm lies a reversible modular exponentiation operator, which, in turn, necessitates the use of reversible operators for the most fundamental arithmetic operations, such as addition and multiplication. Since then, numerous works have been proposed to improve the complexity of Shor's algorithm, focusing in particular on improving the complexity of these arithmetic subroutines \cite{VBE96,TK06,Gid18}. These arithmetic subroutines have also found applications in various quantum algorithms for solving other problems \cite{Reg04,CWM12,LZX20}. In this paper, our focus will be on what is arguably the most fundamental of arithmetic operations: addition.

There are two main methodologies for computing the sum of two $n$-bit numbers contained in quantum registers. The first is derived from classical reversible circuits and thus uses only classical gates (the $\Xg$, $\CNOT$, and Toffoli gates). The second uses the quantum Fourier transform ($\QFT$ for short), which allows to write the sum in the phases by means of rotations, instead of directly working on the quantum registers \cite{Dra02}. It has the significant advantages of requiring no ancilla qubit and running in logarithmic time in its approximate version when only the part performing the addition in the Fourier domain is used, at the cost of relying on inherently quantum gates (such as the Hadamard gate) and small-angle rotation gates. Fourier-based addition is therefore deficient in two major aspects \cite{HRS17}. Firstly, it cannot be efficiently simulated on classical computers, which can slow down implementation, testing and debugging. Secondly and more importantly, it incurs a significant rotation synthesis overhead when quantum error correction is being taken into consideration. Classical reversible arithmetic therefore appears to be more advantageous.

If we take a closer look at these classical reversible methods for addition, two stand out. The first is known as the carry-lookahead method \cite{DKR06} and consists in strongly parallelizing the calculation of the carries occurring in the process of addition, at the cost of using $\Theta(n/\log_2(n))$ ancilla qubits to store intermediate results \cite{TK08,TTK10}. While it achieves logarithmic time complexity with a linear number of gates, its space overhead makes it less practical in cases where the number of qubits at disposal is limited, such as for near-term quantum devices.

Finally, the second classical reversible method and probably the simplest and most intuitive method, known as the ripple-carry technique \cite{VBE96,CDKM04,TTK10}, essentially consists in recursive calculation of the successive carries. It is possible to implement it without any ancilla qubits using a linear number of gates, but with a linear depth, due to the recursive calculation performed involving ladders of $\CNOT$ and Toffoli gates \cite{TTK10}. We give in Figure~\ref{circ:TTK} an example of circuit for the addition of two $n$-bit numbers with $n=5$ derived from the technique proposed in \cite{TTK10}. The linear depth clearly comes from the ladders of $\CNOT$ gates (in brown) and Toffoli gates (in purple) in Slices 2, 3, 5 and 6 (slice numbers refer to the block preceding them). Substitution of these $\CNOT$ and Toffoli ladders with shallower circuits would directly imply an improvement over the linear depth of Takahashi $\etal$ ripple-carry addition circuit, as well as any other circuit based on the ripple-carry technique.

\paragraph{Our contributions.} We give in Section~\ref{sec:preli} some notation and preliminaries before going into the technical details in the following sections. Namely,

\begin{itemize}
    \item In Section~\ref{sec:ladder1}, we prove that one can replace any ladder of $n$ $\CNOT$ gates ($\ie$, a circuit with a $\CNOT$-depth of $n$) by an equivalent $\CNOT$ circuit with logarithmic depth and linear size. We provide the pseudocode for constructing such a circuit and prove its correctness.
    \item In Section~\ref{sec:ladder2}, we show that a ladder of Toffoli gates can be similarly replaced by a logarithmic-depth circuit composed of multi-controlled $\Xg$ gates. In combination with the recent work of Khattar and Gidney \cite{KG24} on the implementation of these multi-controlled $\Xg$ gates in logarithmic depth over the $\{\Xg,\text{Toffoli}\}$ gate set (thanks to the technique of \textit{conditionally clean ancillae}), we prove that any ladder of $n$ Toffoli gates ($\ie$, a circuit with a Toffoli-depth of $\bigO{n}$) can be replaced by a polylogarithmic-depth circuit (more precisely, with a Toffoli-depth of $\bigO{\log^2 n}$) comprising a linearithmic number of Toffoli gates, that does not necessitate any ancilla qubit. We provide the pseudocode for constructing such a circuit, and we prove its correctness.
    \item In Section~\ref{sec:adder}, by applying the results of the two previous sections to a slightly tweaked version of the quantum ripple-carry adder proposed by Takahashi $\etal$ \cite{TTK10}, we ultimately prove that the addition of two $n$-bit numbers on a quantum computer can be done with only Toffoli, $\CNOT$, and $\Xg$ gates, in $\bigO{\log^2 n}$ depth, with $\bigO{n \log n}$ gates and no ancilla. This is the first quantum adder with classical reversible logic only, no ancilla, and $o(n)$ depth (see Table~\ref{table:Adders}).
    \item In Section~\ref{sec:ctrl_adder}, we extend the result of the previous section and prove it holds when considering a controlled version of our quantum ripple-carry adder. This is the operator typically employed in routines such as the modular exponentiation one, which is utilized in Shor's algorithm.
    \item In Section~\ref{sec:QCarith}, we show that the results on the ladders also allow for depth optimization of incrementors and more generally of classical-quantum adders.
\end{itemize}

\begin{figure}[ht]
    \centering
    \includegraphics[width=.9\textwidth]{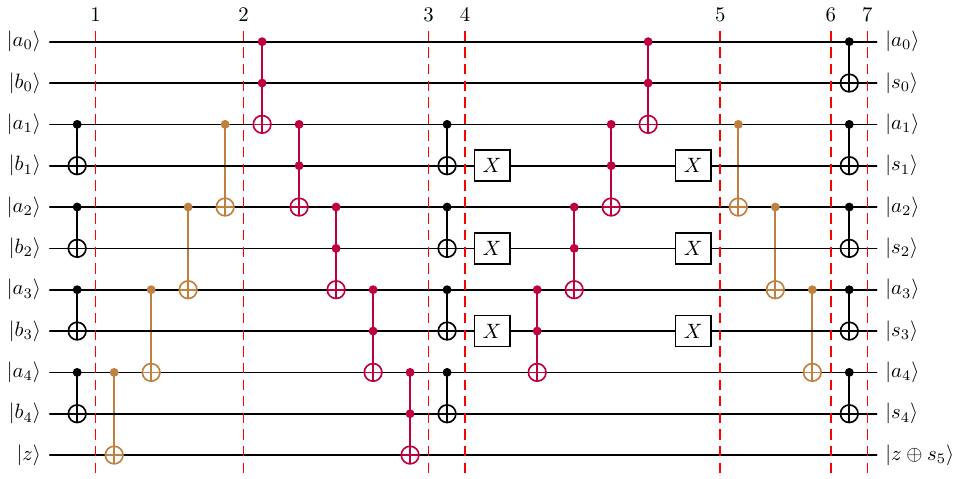}
    \caption{Ancilla-free adder represented as a circuit for $n=5$, derived from \cite{TTK10}. We define $s \eqdef a+b$.}\label{circ:TTK}
\end{figure}

\begin{table}[ht]
    \hspace{1cm}
    \centering
    \begin{NiceTabular}{c||c|c|c|c}[hvlines, first-col]
        & Addition Technique                                        & Size     & Depth 	& Ancilla   \\ \hline\hline
        \Block[c]{1-1}{Fourier-\\[-.3ex]based} 
        & $\QFT$ \cite{Dra02}        & $\bigO{n^2}$ & $\bigO{n}$  & 0       \\\hline\hline
        \Block[c]{3-1}{Classical\\[-.3ex]reversible} 
        & Carry-lookahead \cite{TK08}                   & $\bigO{n}$        & $\bigO{\log n}$   & $\Theta(n / \log n)$ \\
        & Ripple-carry \cite{TTK10}                     & $\bigO{n}$        & $\bigO{n}$        & 0 \\
        & Ripple-carry, Section~\ref{sec:adder}         & $\bigO{n \log n}$ & $\bigO{\log^2 n}$ & 0 \\ \hline
        \CodeAfter
        \SubMatrix\{{3-1}{5-1}.[left-xshift=2mm]
    \end{NiceTabular}
    \caption{Asymptotic complexity of quantum adders. }\label{table:Adders}
\end{table}
\section{Preliminaries} \label{sec:preli}

\subsection{Notation}

The controlled NOT gate, denoted by $\CNOT$, is a well-known quantum gate that operates on two qubits. One of the qubits serves as a control for the application of an $\Xg$ gate on the second qubit, which is referred to as the target. In a similar manner, the well-known Toffoli gate operates on three qubits and involves two control qubits. The application of the $\Xg$ gate to the target is conditional upon the state of both control qubits, which must be set to $\ket{1}$. This notion of controlled $\Xg$ gate can be generalized to an arbitrary number $n$ of control qubits, resulting in what is referred to as the $\MCX_n$ gate (an abbreviation for Multi-Controlled $\Xg$).

\begin{definition}
    Let $t \in \{0,1\}$ and $x_i \in \{0,1\} \; \forall i \in [\![0,n-1]\!]$. We define the multi-controlled $\Xg$ gate with $n$ controls as the operator $\MCX_n$ with the following action:
    \begin{equation}\label{eq:MCX}
        \MCX_n \ket{x_0, \ldots, x_{n-1}, t} = \lvert x_0, \ldots, x_{n-1}, t \oplus \prod_{i=0}^{n-1} x_i\rangle
    \end{equation}
\end{definition}

\noindent Note that we have $\CNOT = \MCX_1$ and $\text{Toffoli} = \MCX_2$.

\noindent Throughout this document, $\log (x)$ will denote the binary logarithm of $x$.

\subsection{Fan-Out operator}

We define the $\fanout_1$ operator as follows:

\begin{definition}
    Let $c \in \{0,1\}$ and $x_i \in \{0,1\} \; \forall i \in [\![0,n-1]\!]$. We define $\fanout_1$ on $n+1$ qubits as the operator $\Fg_1^{(n)}$ with the following action:
    \begin{equation}\label{eq:fanout}
        \Fg_1^{(n)} \Rbra{\ket{c} \otimes \bigotimes_{i=0}^{n-1} \ket{x_i}} \eqdef \ket{c} \otimes \Rbra{\bigotimes_{i=0}^{n-1} \ket{x_i \oplus c}}
    \end{equation}
\end{definition}

A naive implementation of this operator consists of successively applying $\CNOT$ gates using $\ket{c}$ as the control qubit and the $\ket{x_i}$ as successive targets. This implementation requires a total of $n$ $\CNOT$ gates and a depth of $n$. However, it is well known in the literature that this operator can be implemented in logarithmic depth while retaining a linear number of gates and without the use of any ancilla qubit thanks to a divide-and-conquer approach \cite{FFG06,BK09}. 

\begin{lemma}[Folklore]\label{th:fanout}
    The $\Fg_1^{(n)}$ operator can be implemented with only $\CNOT$ gates in depth $\bigO{\log n}$ and size $\bigO{n}$, without ancilla.
\end{lemma}

As its name suggests, the $\fanout_1$ operator processes a number of bits and fans out one of them, designated as the control, by performing a XOR operation with each of the others. This operator can be simply seen as a wall of $\Xg$ gates, whose execution is conditioned by the control bit.

In Section~\ref{sec:ctrl_adder}, we will need to consider an extended version of this operator, namely by considering a wall of $\CNOT$ gates in place of the wall of $\Xg$ gates. That is to say, the question that is posed here is whether a circuit of logarithmic depth can be constructed to implement the following operator, which appears to be novel in the literature:

\begin{definition}
    Let $c \in \{0,1\}$ and $x_i, y_i \in \{0,1\} \; \forall i \in [\![0,n-1]\!]$. We define $\fanout_2$ on $2n+1$ qubits as the operator $\Fg_2^{(n)}$ with the following action:
    \begin{equation}\label{eq:fanout2}
        \Fg_2^{(n)} \Rbra{\ket{c} \otimes \bigotimes_{i=0}^{n-1} \ket{x_i} \ket{y_i}} \eqdef \ket{c} \otimes \Rbra{\bigotimes_{i=0}^{n-1} \ket{x_i} \ket{y_i \oplus c x_i}}
    \end{equation}
\end{definition}

A naive implementation would indeed consist in successively applying $n$ Toffoli gates, all of which share one same control bit (hence the necessity of applying them successively).

\subsection{MCX ladders}

We call $\ladder_1$ the operator which can naively be implemented by means of a $\CNOT$ ladder. This operator is a building block in quantum arithmetic within ripple-carry adders \cite{CDKM04,TTK10}, which are in turn essential components of numerous quantum algorithms \cite{Sho97,Reg04,CWM12,LZX20}.
This ladder of $\CNOT$ gates also appears in quantum circuits that perform binary field multiplication modulo some irreducible primitive polynomials~\cite{Van25}.
Similarly, we call $\ladder_2$ the operator which can naively be implemented by means of a ladder of Toffoli gates and which can also be found in ripple-carry adders. 
We properly define these two operators in the following and discuss how they relate to the $\fanout_1$ operator and to the implementation of $\MCX$ gates.

\subsubsection{CNOT ladder.}

Let us properly define $\ladder_1$ first. 

\begin{definition}
    Let $x_i \in \{0,1\} \; \forall i \in [\![0,n]\!]$. We define $\ladder_1$ on $n+1$ qubits as the operator $\Lg_1^{(n)}$ with the following action:
    \begin{equation}\label{eq:ladder1}
        \Lg_1^{(n)} \Rbra{\bigotimes_{i=0}^{n} \ket{x_i}}  \eqdef \ket{x_0} \otimes \Rbra{\bigotimes_{i=1}^{n} \ket{x_i \oplus x_{i-1}}}
    \end{equation}
\end{definition}

A naive implementation of this operator employs a linear number of sequentially applied $\CNOT$ gates (as shown for example in Figure~\ref{circ:TTK}: $\Lg_1^{(4)}$ in Slice 2 and $\Rbra{\Lg_1^{(3)}}^\dagger$ in Slice 6), resulting in a linear depth in terms of $\CNOT$ gates. However, to the best of our knowledge, this operator has never been studied in the literature. Section~\ref{sec:ladder1} fills this gap.

\paragraph{On the link between Fan-out and Ladder$_1$.}

The $\fanout_1$ operator and $\ladder_1$ are related by the following equation:
\begin{equation}\label{eq:fanout_ladder}
    \Fg_1^{(n)} = \Rbra{\Ig \otimes \Lg_1^{(n-1)}}^\dagger \circ \Lg_1^{(n)} 
\end{equation}
where $\Ig$ denotes the identity.
It follows from this equality that there is a circuit implementing $\Fg_1^{(n)}$ that costs no more than twice the implementation cost of $\Lg_1^{(n)}$.
While the $\fanout_1$ operator can be implemented by applying a constant number of $\ladder_1$ operators via Equation~\ref{eq:fanout_ladder}, the converse does not hold.
Indeed, over $n$ qubits, the $\ladder_1$ operator acts non-trivially on $n-1$ qubits, whereas the $\fanout_1$ operator acts non-trivially on only $1$ qubit.
This implies that the $\ladder_1$ operator cannot be implemented by applying only a constant number of $\fanout_1$ operators.
Despite this, we show in Section~\ref{sec:ladder1} that the $\ladder_1$ operator can be implemented with a logarithmic depth complexity, which matches the depth complexity of the $\fanout_1$ operator.

\subsubsection{Toffoli ladder and generalization.}

We now generalize the definition of $\Lg_1$ and define the $\Lg_2$ operator, corresponding to the action of a Toffoli ladder. 

\begin{definition}
    Let $x_i \in \{0,1\} \; \forall i \in [\![0,2n]\!]$. We define $\ladder_2$ on $2n+1$ qubits as the operator $\Lg_2^{(n)}$ with the following action:
    \begin{equation}\label{eq:ladder2}
        \Lg_2^{(n)} \Rbra{\bigotimes_{i=0}^{2n} \ket{x_i}} \eqdef \ket{x_0} \otimes \Rbra{\bigotimes_{i=1}^{n} \ket{x_{2i-1}} \ket{x_{2i} \oplus x_{2i-2}x_{2i-1}}}.
    \end{equation}
\end{definition}

In the literature, no references seem to exist that discuss another way of implementing the $\Lg_2^{(n)}$ operator than the straightforward linear-depth manner of sequentially applying $n$ Toffoli gates (as shown for example in Figure~\ref{circ:TTK}; $\Rbra{\Lg_2^{(5)}}^\dagger$ in Slice 3 and $\Lg_2^{(4)}$ in Slice 5). That said, it is noteworthy that Toffoli ladders appear quite naturally when attempting to replace a multi-controlled $\Xg$ gate with a circuit involving only Toffoli gates (see Gidney's website \cite{Gid15}). In other words, the $\ladder_2$ operator is closely related to the $\MCX$ gates. Finally, it is worth noting once again that these Toffoli ladders appear in addition circuits.

More generally, we can define $\Lg_{\bs \alpha}$ as an operator associated with a ladder composed of multiple $\MCX$ gates. 
Here, $\bs \alpha$ is a vector of integers that specifies the control and target qubits in the ladder. 
Specifically, the $i$-th $\MCX$ gate in the ladder has $\alpha_i$ as its target qubit and is controlled by all the qubits between $\alpha_{i-1}$ and $\alpha_i$.
\begin{definition}
    Let $\bs\alpha$ be a vector of $k-1$ integers satisfying $0 \leq \alpha_0 < \alpha_1 < \ldots < \alpha_{k-2}$, then we call $\ladder_{\bs\alpha}$ on $\alpha_{k-2}+1$ qubits the operator:
    \begin{equation}\label{eq:ladder_alpha}
        \Lg_{\bs\alpha} \ket{\bs x} = \left(\bigotimes_{i=0}^{\alpha_0-1} \ket{x_{i}}\right) \lvert x_{\alpha_{0}} \oplus \prod_{j=0}^{\alpha_0 - 1} x_j \rangle \bigotimes_{i=1}^{k-2} \left(\bigotimes_{j=\alpha_{i-1}+1}^{\alpha_i-1} \ket{x_{j}}\right) \lvert x_{\alpha_{i}} \oplus \prod_{j=\alpha_{i-1}}^{\alpha_i - 1} x_j \rangle.
    \end{equation}
\end{definition}
Note that $\Lg_{\bs \alpha} = \Lg^{(n)}_{1}$ in the case where $\bs \alpha = (1, 2, 3, \ldots, n)$, and  $\Lg_{\bs \alpha} = \Lg^{(n)}_{2}$ in the case where $\bs \alpha = (2, 4, 6, \ldots, 2n)$.

\subsubsection{Logarithmic-depth implementation of MCX gates.}

It has recently been shown by Nie $\etal$ \cite{NZS24} that $\MCX_n$ gates can be implemented in logarithmic depth (in $n$) with one clean ancilla. Khattar and Gidney \cite{KG24} then extended this work and showed that we can implement $\MCX_n$ gates with a circuit with $\bigO{n}$ Toffoli gates and $\bigO{\log n}$ depth using no more than two dirty ancillae. These works share the same idea of what is called the technique of \textit{conditionally clean ancillae}. We specifically recall in Theorem~\ref{th:MCX} a result of \cite{KG24} that we will use later in this paper.

\begin{theorem}[Section 5.4 in \cite{KG24}]\label{th:MCX} 
    The $\MCX_n$ operator can be implemented with only Toffoli and $\Xg$ gates in depth $\bigO{\log n}$ and size $\bigO{n}$, with two dirty ancilla qubits.
\end{theorem}

This result will be crucial in Section~\ref{sec:ladder2}, where it will be employed to prove that we can implement $\Lg_2^{(n)}$ with a circuit that has a Toffoli-depth in $\bigO{\log^2 n}$, $\bigO{n \log n}$ Toffoli gates and no ancilla. 
\section{Logarithmic-depth implementation of the CNOT ladder operator}\label{sec:ladder1}

In this section, we present a logarithmic-depth implementation of the $\CNOT$ ladder operator, introduced in Section~\ref{sec:preli} and denoted by $\Lg_1^{(n-1)}$ over $n$ qubits.
Consider the pseudocode presented in Algorithm~\ref{algo:ladder1}, in which $::$ denotes the concatenation operator.

\begin{algorithm}[!ht]
\caption{Logarithmic-depth $\CNOT$ circuit synthesis for $\Lg_1^{(n-1)}$}\label{algo:ladder1}
\begin{algorithmic}[1]
\Require A list $X=X_0, \ldots, X_{n-1}$ of $n$ qubits.
\Ensure A circuit implementing the $\ladder_1$ operator over the qubits $X$.
\Procedure{Ladder$_1$-Synth}{$X$}
    \If{$n = 1$}
        \State \Return empty circuit
    \EndIf
    \If{$n = 2$}
        \State \Return $\CNOT\Rbra{X_0, X_1}$
    \EndIf
    \State $X' \gets X_1$
    \State $C_R \gets \CNOT\Rbra{X_0, X_1}$
    \State $C_L \gets \CNOT\Rbra{X_{n-2}, X_{n-1}}$
    \For{$i=1$ to $\ceil{n/2}-2$}
        \State $C_L \gets C_L :: \CNOT\Rbra{X_{2i - 1}, X_{2i}}$
        \State $C_R \gets C_R :: \CNOT\Rbra{X_{2i}, X_{2i + 1}}$
        \State $X' \gets X' :: X_{2i+1}$
    \EndFor
    \If{$n$ is even}
        \State $X' \gets X' :: X_{n-2}$
    \EndIf
    \State \Return $C_L$ :: \Call{Ladder$_1$-Synth}{$X'$} :: $C_R$
\EndProcedure
\end{algorithmic}
\end{algorithm}

Algorithm~\ref{algo:ladder1} constructs two $\CNOT$ circuits of depth 1, denoted by $C_L$ and $C_R$, and performs a recursive call on $\floor{n/2}$ qubits (denoted by $X'$) to produce a subcircuit that is inserted between $C_L$ and $C_R$.
For illustration, Figure~\ref{circ:ladder1} provides an example of the $\CNOT$ circuit resulting from this procedure when $n=10$.

\begin{figure}[!ht]
    \centering
    \includegraphics[width=.9\textwidth]{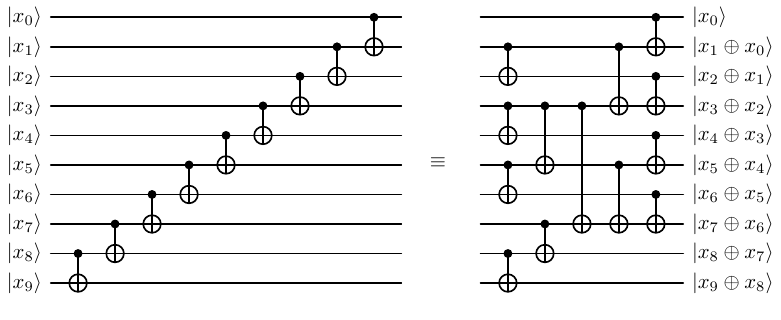}
    \caption{On the left, a linear-depth circuit implementing the $\Lg_1^{(9)}$ operator. On the right, an equivalent logarithmic-depth circuit produced by Algorithm~\ref{algo:ladder1}.} \label{circ:ladder1}
\end{figure}

The exact $\CNOT$-depth and $\CNOT$-count of the circuit produced by Algorithm~\ref{algo:ladder1} are stated in Lemma~\ref{thm:ladder1}.

\begin{lemma}\label{thm:ladder1}
    Let $n \ge 2$ be an integer. The circuit produced by Algorithm~\ref{algo:ladder1} implements $\Lg_1^{(n-1)}$ with a $\CNOT$-depth of 
    \[
	\floor{\log n} + \floor{\log \frac{2n}{3}} \qquad (\le 2\floor{\log n})
    \] 
    and a $\CNOT$-count of
    \[
        2n - 2 - \floor{\log n} - \floor{\log \frac{2n}{3}}.
    \]
\end{lemma}

\begin{proof}
    We first prove that the circuit produced by Algorithm~\ref{algo:ladder1} implements $\Lg_1^{(n-1)}$.
    For the base case where $n = 1$, $\Lg_1^{(0)}$ is equal to the identity operator, which corresponds to the empty circuit returned by Algorithm~\ref{algo:ladder1}.
    For the other base case where $n = 2$, the algorithm produces a circuit containing a single $\CNOT$ gate, which corresponds to the implementation of the $\Lg_1^{(1)}$ operator:
    \begin{equation}
        \CNOT \ket{x_0}\ket{x_1} = \ket{x_0}\ket{x_0 \oplus x_1} = \Lg_1^{(1)}\ket{x_0}\ket{x_1}.
    \end{equation}
    For the other cases where $n > 2$, Algorithm~\ref{algo:ladder1} constructs two circuits, $C_L$ and $C_R$, and performs a recursive call with a subset of qubits $X'$, which produces a circuit that we denote by $C_{X'}$.
    The circuit produced by Algorithm~\ref{algo:ladder1} is then the result of the concatenation of the circuits $C_L$, $C_{X'}$, and $C_R$.
    Let $\Ug_L$, $\Ug_{X'}$, and $\Ug_R$ be the unitary operators associated with the circuits $C_L$, $C_{X'}$, and $C_R$, respectively.
        The action of the $\Ug_L$ operator on an $n$-dimensional computational basis state $\ket{\bs x}$ is as follows:
    \begin{equation*}
        \Ug_L\ket{\bs x} = \ket{x_0}\ket{x_1} \left(\bigotimes_{i=1}^{\ceil{\frac{n}{2}}-2} \ket{x_{2i-1} \oplus x_{2i}}\ket{x_{2i+1}}\right) \left(\bigotimes_{n\bmod{2}}^{0} \ket{x_{n-2}}\right) \ket{x_{n-2} \oplus x_{n-1}}
    \end{equation*}
    The operator $\Ug_{X'}$ implements the $\Lg^{(\floor{n/2}-1)}_1$ operator on the qubits $X'$, which corresponds to the sequence of qubits $X_{2i+1}$ for all $i$ satisfying $0 \leq i \leq \ceil{\frac{n}{2}}-2$, followed by $X_{n-2}$ when $n$ is even.
    Thus, the action of the $\Ug_{X'}$ operator on an $n$-dimensional computational basis state $\ket{\bs x}$ is as follows:
    \begin{equation*}
        \Ug_{X'}\ket{\bs x} = \ket{x_{0}}\ket{x_{1}} \left(\bigotimes_{i=1}^{\ceil{\frac{n}{2}}-2} \ket{x_{2i}}\ket{x_{2i-1} \oplus x_{2i+1}}\right) \left(\bigotimes_{n\bmod{2}}^{0} \ket{x_{n-3} \oplus x_{n-2}}\right) \ket{x_{n-1}}
    \end{equation*}
    And the action of the $\Ug_R$ operator on an $n$-dimensional computational basis state $\ket{\bs x}$ is as follows:
    \begin{equation*}
        \Ug_R\ket{\bs x} = \ket{x_{0}}\ket{x_{0} \oplus x_{1}} \left(\bigotimes_{i=1}^{\ceil{\frac{n}{2}}-2} \ket{x_{2i}}\ket{x_{2i} \oplus x_{2i+1}}\right) \left(\bigotimes_{n\bmod{2}}^{0} \ket{x_{n-2}}\right) \ket{x_{n-1}}
    \end{equation*}
    By putting these equations together, the operation performed by the $\Ug_R \Ug_{X'} \Ug_L$ operator can be described as follows:
    \begin{equation*}
        \begin{aligned}
            &\Ug_R\Ug_{X'}\Ug_L\ket{\bs x} \\
            &= \Ug_R \Ug_{X'} \left[\ket{x_0}\ket{x_1} \left(\bigotimes_{i=1}^{\ceil{\frac{n}{2}}-2} \ket{x_{2i-1} \oplus x_{2i}} \ket{x_{2i+1}}\right) \left(\bigotimes_{n\bmod{2}}^{0} \ket{x_{n-2}}\right) \ket{x_{n-2} \oplus x_{n-1}}\right] \\
            &= \Ug_R \left[\ket{x_{0}}\ket{x_{1}} \left(\bigotimes_{i=1}^{\ceil{\frac{n}{2}}-2} \ket{x_{2i-1} \oplus x_{2i}} \ket{x_{2i-1} \oplus x_{2i+1}}\right) \left(\bigotimes_{n\bmod{2}}^{0} \ket{x_{n-3} \oplus x_{n-2}}\right) \ket{x_{n-2} \oplus x_{n-1}}\right] \\
            &= \ket{x_{0}}\ket{x_{0} \oplus x_{1}} \left(\bigotimes_{i=1}^{\ceil{\frac{n}{2}}-2} \ket{x_{2i-1} \oplus x_{2i}} \ket{x_{2i} \oplus x_{2i+1}}\right) \left(\bigotimes_{n\bmod{2}}^{0} \ket{x_{n-3} \oplus x_{n-2}}\right) \ket{x_{n-2} \oplus x_{n-1}} \\
            &= \ket{x_{0}}\bigotimes_{i=1}^{n-1} \ket{x_i \oplus x_{i-1}} \\
            &= \Lg_1^{(n-1)}\ket{\bs x}
        \end{aligned}
    \end{equation*}
    Thus, the circuit produced by Algorithm~\ref{algo:ladder1}, associated with the operator $\Ug_R\Ug_{X'}\Ug_L$, produces a circuit implementing the $\Lg_1^{(n-1)}$ operator.

    The $\CNOT$-depth of the $\Ug_L$ and $\Ug_R$ circuits is exactly one, because all the $\CNOT$ gates in the circuits $C_L$ and $C_R$ are applied on different qubits.
    Therefore, the depth $D(n)$ of the circuit produced by Algorithm~\ref{algo:ladder1} is
    \begin{equation}\label{eq:depth_ladder1}
        D(n) = 2 + D(\floor{n/2})
    \end{equation}
    with the initial conditions $D(2) = 1$ and $D(3)=2$.
    Solving this recurrence relation yields
    \begin{equation}
        D(n) = \floor{\log{n}} + \floor{\log{\frac{2n}{3}}}.
    \end{equation}
    Similarly, regarding the cost, the number of $\CNOT$ gates in the $C_L$ and $C_R$ circuits is
    \begin{equation}
        \floor{\frac{n-1}{2}},
    \end{equation}
    which implies that the number of $\CNOT$ gates in the circuit produced by Algorithm~\ref{algo:ladder1} is
    \begin{equation}\label{eq:cnot_ladder1}
        C(n) = 2\floor{\frac{n-1}{2}} + C\Rbra{\floor{\frac{n}{2}}}
    \end{equation}
    with the initial conditions $C(2) = 1$ and $C(3)=2$.
    Solving this recurrence relation yields
    \begin{equation}
        C(n) = 2 n - 2 - \floor{\log{n}} - \floor{\log{\frac{2n}{3}}}.
    \end{equation}
\end{proof}
\section{Polylogarithmic-depth implementation of the Toffoli ladder operator}\label{sec:ladder2}

In this section, we present a polylogarithmic-depth implementation of the Toffoli ladder operator, introduced in Section~\ref{sec:preli} and denoted by $\Lg_2^{(n)}$ over $2n+1$ qubits.

We begin by presenting an algorithm for implementing the $\Lg_{\bs \alpha}$ operator (also introduced in Section~\ref{sec:preli}) with logarithmic depth using $\MCX$ gates.
Consider the pseudocode presented in Algorithm~\ref{algo:ladder_alpha}, where $::$ denotes the concatenation operator.

\begin{algorithm}[!ht]
\caption{Logarithmic-depth $\MCX$ circuit synthesis for $\Lg_{\bs \alpha}$} \label{algo:ladder_alpha}
\begin{algorithmic}[1]
\Require A vector $\bs \alpha$ of $k-1$ integers associated with the $\Lg_{\bs \alpha}$ operator, and a list $X=X_0, \ldots, X_{\alpha_{k-2}}$ of $\alpha_{k-2}+1$ qubits.
\Ensure A circuit implementing the $\ladder_{\bs \alpha}$ operator over the qubits $X$.
\Procedure{Ladder$_\alpha$-Synth}{$X$, $\bs \alpha$}
    \If{$k = 1$}
        \State \Return empty circuit
    \EndIf
    \If{$k = 2$}
        \State \Return $\MCX\Rbra{X_{0}, \ldots, X_{\alpha_0}}$ \Comment{$\MCX$ gate with controls $X_{i}$ where $0 \leq i < \alpha_0$, and target $X_{\alpha_0}$.}
    \EndIf
    \State $X' \gets X_{\alpha_0}$
    \State $\bs \alpha' \gets \text{empty vector of integers}$
    \State $C_R \gets \MCX\Rbra{X_{0}, \ldots, X_{\alpha_0}}$
    \State $C_L \gets \MCX\Rbra{X_{\alpha_{k-3}}, \ldots, X_{\alpha_{k-2}}}$
    \For{$i=1$ to $\ceil{k/2}-2$}
        \State $C_L \gets C_L :: \MCX\Rbra{X_{\alpha_{2i-2}}, \ldots, X_{\alpha_{2i-1}}}$
        \State $C_R \gets C_R :: \MCX\Rbra{X_{\alpha_{2i-1}}, \ldots, X_{\alpha_{2i}}}$
        \State $X' \gets X' :: \left[X_{\alpha_{2i-2}+1}, \ldots, X_{\alpha_{2i-1}-1}, X_{\alpha_{2i-1}+1}, \ldots, X_{\alpha_{2i}} \right]$
        \State $\bs \alpha' \gets \bs \alpha' :: \alpha_{2i} - \alpha_0 - i$
    \EndFor
    \If{$k$ is even}
        \State $X' \gets X' :: \left[X_{\alpha_{k-4}+1}, \ldots, X_{\alpha_{k-3}} \right]$
        \State $\bs \alpha' \gets \bs \alpha' :: \alpha_{k-3} - \alpha_0 - k/2 + 2$
    \EndIf
    \State \Return $C_L$ :: \Call{Ladder$_\alpha$-Synth}{$X'$, $\bs \alpha'$} :: $C_R$
\EndProcedure
\end{algorithmic}
\end{algorithm}

Analogously to Algorithm~\ref{algo:ladder1}, the algorithm constructs two $\MCX$ circuits of depth 1, denoted by $C_L$ and $C_R$, and performs a recursive call to produce a subcircuit that is inserted between $C_L$ and $C_R$.
For illustration, Figure~\ref{circ:ladder2} provides an example of the CNOT circuit resulting from this procedure when $\bs \alpha = (2, 4, \ldots, 18)$.

\begin{figure}[!ht]
    \centering
    \includegraphics[width=.9\textwidth]{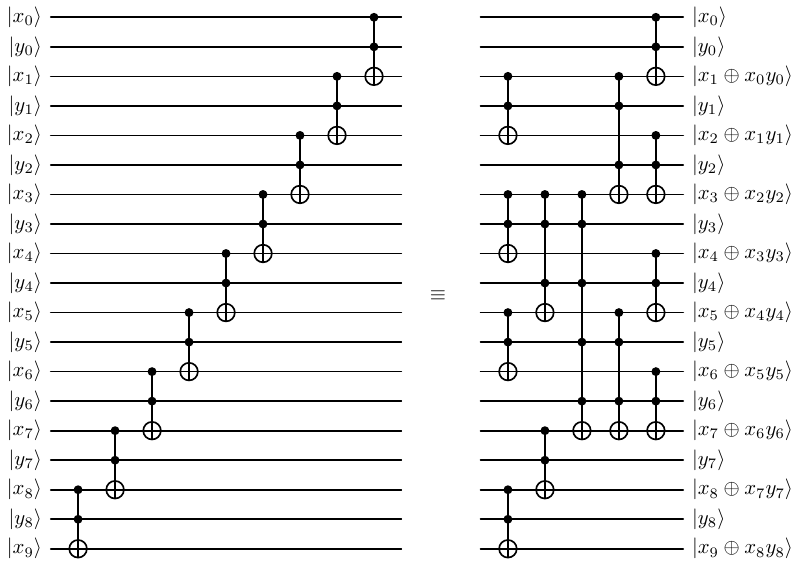}
    \caption{On the left, a linear-depth Toffoli circuit implementing the $\Lg_2^{(9)}$ operator. On the right, an equivalent logarithmic-depth $\MCX$ circuit produced by Algorithm~\ref{algo:ladder_alpha}.} \label{circ:ladder2}
\end{figure}

The exact $\MCX$-depth and $\MCX$-count of the circuit produced by Algorithm~\ref{algo:ladder1} are stated in Lemma~\ref{thm:ladder_alpha}.

\begin{restatable}{lemma}{thmladderalpha} \label{thm:ladder_alpha}
    Let $\bs \alpha$ be a vector of $k-1$ integers, where $k\geq 2$, associated with the $\Lg_{\bs \alpha}$ operator.
    The circuit produced by Algorithm~\ref{algo:ladder_alpha} implements $\Lg_{\bs \alpha}$ with a $\MCX$-depth of 
    \[
	\floor{\log \Rbra{k}} + \floor{\log\Rbra{\frac{2k}{3}}} \qquad (\le 2\floor{\log \Rbra{k}})
    \] 
    and a $\MCX$-count of
    \[
        2k - 2 - \floor{\log \Rbra{k}} - \floor{\log\Rbra{\frac{2k}{3}}}.
    \]
\end{restatable}

The proof of Lemma~\ref{thm:ladder_alpha} is similar to the one of Lemma~\ref{thm:ladder1}. \ifarXiv We provide it in Appendix~\ref{sec:proof_ladder_alpha}. \fi

The logarithmic-depth $\MCX$-circuit produced by Algorithm~\ref{algo:ladder_alpha} can be translated into a $\{\text{Toffoli}, \Xg\}$ circuit by using Theorem~\ref{th:MCX}.
For the $\Lg^{(n)}_2$ operator, we obtain a polylogarithmic-depth circuit with a linearithmic number of Toffoli gates, as stated in Lemma~\ref{thm:ladder2}. 

\begin{lemma}\label{thm:ladder2}
    There exists a circuit that implements $\Lg_2^{(n)}$ over the $\{\text{Toffoli}, \Xg\}$ gate set with a depth of $\bigO{\log^2 n}$ and a gate count of $\bigO{n \log n}$, without any ancilla qubits.
\end{lemma}
\begin{proof}
    The first and last layers of the circuit are composed of parallel Toffoli gates, inducing a depth of $2$ and a number of Toffoli gates of $\bigO{n}$. The first two qubits of the circuit are not used in any other layer of the circuit. Moreover, for all the other layers of the circuit, the parallel $\MCX$ gates are all separated by at least two qubits on which no gates are acting. Therefore, for each one of these layers, two different dirty ancillary qubits can be associated to each $\MCX$ gate. Then, based on Theorem~\ref{th:MCX}, we can implement all the $\MCX$ gates in a given layer in parallel over the $\{\text{Toffoli}, \Xg\}$ gate set, with a depth of $\bigO{\log m_i}$ and a gate count of $\bigO{m_i}$ for each gate, where $m_i$ is the number of controls of the $i$-th $\MCX$ gate in the layer.
    We have $\max_i(m_i) \leq n$, which implies that the total depth of the layer is $\bigO{\log n}$, as all the $\MCX$ gates are implemented in parallel.
    Moreover, we have $\sum_i m_i \leq n$, which implies that the total number of $\{\text{Toffoli}, \Xg\}$ gates in the layer is $\bigO{n}$.
    As stated by Lemma~\ref{thm:ladder_alpha}, there are $\bigO{\log n}$ layers of parallel $\MCX$ gates in the initial circuit, which results in a $\{\text{Toffoli}, \Xg\}$ circuit with a depth complexity of $\bigO{\log^2 n}$ and a gate count of $\bigO{n \log n}$.
\end{proof}
\section{Application to Ripple-Carry Addition}\label{sec:adder}

We can now rely on the results established in the previous sections to introduce the main result of this paper: a polylogarithmic-depth and ancilla-free quantum adder using classical reversible logic only.

The ripple-carry adder presented in Algorithm~\ref{algo:adder} is derived from Takahashi's $\etal$ ripple-carry adder \cite{TTK10} by applying the following circuit equality on their original adder:
\begin{equation}
    \begin{quantikz}[align equals at=2, row sep={.5cm,between origins}]
        \qw & \ctrl{2} & \ctrl{1} & \qw \\
        \qw & \ctrl{1} & \targ{} & \qw \\
        \qw & \targ{} & \qw & \qw
    \end{quantikz}
    =
    \begin{quantikz}[align equals at=2, row sep={.5cm,between origins}]
        \qw & \ctrl{1} & \qw & \ctrl{2} & \qw & \qw \\
        \qw & \targ{} & \gate{X} & \ctrl{1} & \gate{X} & \qw \\
        \qw & \qw & \qw & \targ{} & \qw & \qw
    \end{quantikz}
\end{equation}
An example of the circuit produced by Algorithm~\ref{algo:adder} for $n=5$ is provided in Figure~\ref{circ:TTK}, where linear-depth implementations of the $\CNOT$ and Toffoli ladder operators are used. The corresponding pseudocode in Algorithm~\ref{algo:adder} is expressed using the $\Lg_1$ and $\Lg_2$ operators.

\begin{algorithm}[ht]
    \caption{Ancilla-free ripple-carry adder}\label{algo:adder}
    \begin{algorithmic}[1]
        \Require $\ket{a}_{A} \ket{b}_{B} \ket{z}_{Z}$ where $a,b \in [\![0,2^n - 1]\!]$ and $z \in \{0,1\}$.
        \Ensure $\ket{a}_{A} \ket{a+b \mod 2^n}_{B} \ket{z \oplus (a+b)_n}_{Z}$.
        \For{$i = 1$ to $n - 1$} \Comment{Slice 1}
            \State $\CNOT(A_i,B_i)$ 
        \EndFor
        \State Apply $\Lg_1^{(n-1)}$ on $(A_1, \dots, A_{n-1}, Z)$ \Comment{Slice 2}
        \State Apply $\Rbra{\Lg_2^{(n)}}^\dagger$ on $(A_0, B_0, \dots, A_{n-1}, B_{n-1}, Z)$ \Comment{Slice 3}
        \For{$i = 1$ to $n-1$} \Comment{Slice 4}
            \State $\CNOT(A_i,B_i)$
        \EndFor
        \For{$i = 1$ to $n-2$} \Comment{Slice 5}
            \State $\Xg(B_i)$
        \EndFor
        \State Apply $\Lg_2^{(n-1)}$ on $(A_0, B_0, \dots, A_{n-2}, B_{n-2}, A_{n-1})$
        \For{$i = 1$ to $n-2$}
            \State $\Xg(B_i)$
        \EndFor
        \State Apply $\Rbra{\Lg_1^{(n-2)}}^\dagger$ on $(A_1, \dots, A_{n-1})$ \Comment{Slice 6}
        \For{$i = 0$ to $n-1$} \Comment{Slice 7}
            \State $\CNOT(A_i,B_i)$
        \EndFor
    \end{algorithmic}
\end{algorithm}

We establish in Theorem~\ref{thm:adder} that by using the logarithmic-depth circuit to replace the $\CNOT$ ladders (Section~\ref{sec:ladder1}) and the polylogarithmic-depth circuit to replace the Toffoli ladders (Section~\ref{sec:ladder2}), ripple-carry addition inherits this same complexity asymptotically.

\begin{theorem}\label{thm:adder}
    Let $a$ and $b$ be two $n$-bit integers. There exists a circuit that implements the in-place addition of $a$ and $b$, $\ie$, an operator with the following action:
    \[
	\ket{a} \ket{b} \ket{z} \mapsto \ket{a} \ket{a + b \mod 2^n} \ket{z \oplus (a + b)_n}
    \]
	 (where $z \in \{0,1\}$) over the $\{\text{Toffoli}, \CNOT, \Xg\}$ gate set with a depth of $\bigO{\log^2 n}$ and a gate count of $\bigO{n \log n}$, without any ancilla qubits.
\end{theorem}
\begin{proof}
    We prove Theorem~\ref{thm:adder} by demonstrating that each slice of the circuit produced by Algorithm~\ref{algo:adder}, which implements the in-place addition between $a$ and $b$, can be implemented over the $\{\text{Toffoli}, \CNOT, X\}$ gate set with depth and gate count complexities of at most $\bigO{\log^2 n}$ and $\bigO{n\log n}$, respectively.
    Slices 1, 4, and 7 are each composed of a single layer of parallel $\CNOT$ gates, inducing a depth of $\bigO{1}$ and a gate count of $\bigO{n}$.
    Slices 2 and 6 can be implemented with depth $\bigO{\log n}$ and size $\bigO{n}$ using only $\CNOT$ gates, as proved in Lemma~\ref{thm:ladder1}.
    Finally, Slices 3 and 5 can be implemented with depth $\bigO{\log^2 n}$ and size $\bigO{n \log n}$ over the $\{\text{Toffoli}, \Xg\}$ gate set as proved in Lemma~\ref{thm:ladder2}.
\end{proof}

\section{Controlled Ripple-Carry Addition}\label{sec:ctrl_adder}

Based on the results established in Section~\ref{sec:adder}, we present a controlled adder with the same asymptotic depth and gate count complexities as the adder produced by Algorithm~\ref{algo:adder}, and which also does not use any ancilla qubits.

We will rely on the following lemma, which states that the $\Fg^{(n)}_2$ operator can be implemented with the same asymptotic depth and gate count complexities as the $\Fg^{(n)}_1$ operator.
\begin{lemma}\label{lem:fanout_2}
    The $\Fg^{(n)}_2$ operator can be implemented over the $\{\text{Toffoli}, \CNOT\}$ gate set with a depth of $\bigO{\log n}$ and a gate count of $\bigO{n}$, without any ancilla qubits.
\end{lemma}
\begin{proof}
    We rely on the following equality, which is for example used in~\cite{NZS24,KG24}:
    \begin{equation}\label{eq:ctrl_u_dirty_ancilla}
        \begin{quantikz}[align equals at=2, row sep={.5cm,between origins}]
            & \qw & \ctrl{2} & \qw \\
            & \qw & \qw & \qw \\
            & \qwbundle{m} & \gate{U} & \qw \\
        \end{quantikz}
        =
        \begin{quantikz}[align equals at=2, row sep={.5cm,between origins}]
            & \qw & \ctrl{1} & \qw & \ctrl{1} & \qw & \qw \\
            & \qw & \targ{} & \ctrl{1} & \targ{} & \ctrl{1} & \qw \\
            & \qwbundle{m} & \qw & \gate{U} & \qw & \gate{U} & \qw \\
        \end{quantikz}
    \end{equation}
    where $U^2 = I$.
    Based on this equality, we can derive the following equality:

    \begin{equation}\label{eq:fan_out_2}
        \begin{quantikz}[column sep=.3cm, row sep={.5cm,between origins}]
            & \qw & \ctrl{5} & \qw \\
            & \qw & \qw & \qw \\
            & \qwbundle{m} & \gate{U} & \qw \\
            & \wave & & \\
            & \qw & \qw & \qw \\
            & \qwbundle{m} & \gate{U} & \qw
        \end{quantikz}
        =
        \begin{quantikz}[column sep=.3cm, row sep={.5cm,between origins}]
            & \qw & \ctrl{4} & \qw & \ctrl{4} & \qw & \qw \\
            & \qw & \targ{} & \ctrl{1} & \targ{} & \ctrl{1} & \qw \\
            & \qwbundle{m} & \qw & \gate{U} & \qw & \gate{U} & \qw \\
            & \wave & & & & & & & \\
            & \qw & \targ{} & \ctrl{1} & \targ{} & \ctrl{1} & \qw \\
            & \qwbundle{m} & \qw & \gate{U} & \qw & \gate{U} & \qw
        \end{quantikz}
    \end{equation}
    where $U^2=I$.
    Notice that in the case where $m=2$ and $U=\CNOT$, this circuit implements the $\Fg^{(n)}_2$ operator using $n$ dirty ancilla qubits.
    \ifarXiv For example, in the case where $n=3$, we get the following circuit equality for the implementation of the $\Fg^{(3)}_2$ operator using $3$ dirty ancilla qubits:
    \begin{equation}
        \begin{quantikz}[column sep=.3cm, row sep={.5cm,between origins}]
            & \ctrl{2} & \ctrl{5} & \ctrl{8} & \\
            & \qw & \qw & \qw & \\
            & \ctrl{1} & \qw & \qw & \\
            & \targ{} & \qw & \qw & \\
            & \qw & \qw & \qw & \\
            & \qw & \ctrl{1} & \qw & \\
            & \qw & \targ{} & \qw  & \\
            & \qw & \qw & \qw & \\
            & \qw & \qw & \ctrl{1} & \\
            & \qw & \qw & \targ{} &
        \end{quantikz}
        =
        \begin{quantikz}[column sep=.3cm, row sep={.5cm,between origins}]
            & \ctrl{7} & \qw & \ctrl{7} & \qw & \\
            & \targ{} & \ctrl{1} & \targ{} & \ctrl{1} & \\
            & \qw & \ctrl{1} & \qw & \ctrl{1} & \\
            & \qw & \targ{} & \qw & \targ{} & \\
            & \targ{} & \ctrl{1} & \targ{} & \ctrl{1} & \\
            & \qw & \ctrl{1} & \qw & \ctrl{1} & \\
            & \qw & \targ{} & \qw & \targ{} & \\
            & \targ{} & \ctrl{1} & \targ{} & \ctrl{1} & \\
            & \qw & \ctrl{1} & \qw & \ctrl{1} & \\
            & \qw & \targ{} & \qw & \targ{} &
        \end{quantikz}
    \end{equation} \fi
    As such, the $\Fg^{(n)}_2$ operator can be implemented by two layers of parallel Toffoli gates and two $\Fg^{(n)}_1$ operators, using $n$ dirty ancilla qubits.

    In the case where $n=1$, the $\Fg^{(n)}_2$ operator can be implemented with a single Toffoli gate.
    In the more general case where $n\geq 2$, the $\Fg^{(n)}_2$ operator can be split into two operators which can be implemented sequentially: $\Fg^{(\ceil{n/2})}_2$ and $\Fg^{(\floor{n/2})}_2$.
    These two operators do not act on at least $\ceil{n/2}$ qubits, which can be used as dirty ancilla qubits to implement the operators as in the right-hand side of Equation~\ref{eq:fan_out_2}.
    This results in a circuit implementing the $\Fg^{(n)}_2$ operator with $4$ layers of parallel Toffoli gates, two $\Fg^{(\ceil{n/2})}_1$ operators and two $\Fg^{(\floor{n/2})}_1$ operators.
    As stated by Lemma~\ref{th:fanout}, the $\Fg^{(\ceil{n/2})}_1$ and $\Fg^{(\floor{n/2})}_1$ operators can be implemented with a depth of $\bigO{\log n}$ using $\bigO{n}$ $\CNOT$ gates.
    Thus, the $\Fg^{(n)}_2$ operator can be implemented with a depth of $\bigO{\log n}$ and a gate count of $\bigO{n}$, without any ancilla qubits.
\end{proof}

\begin{algorithm}[H]
    \caption{Ancilla-free controlled ripple-carry adder}\label{algo:controlled_adder}
    \begin{algorithmic}[1]
        \Require $\ket{c}_{C} \ket{a}_{A} \ket{b}_{B} \ket{z}_{Z}$ where $c \in \{0,1\}$, $a,b \in [\![0,2^n - 1]\!]$ and $z \in \{0,1\}$.
        \Ensure $\ket{c}_{C} \ket{a}_{A} \ket{ca + b \mod 2^n}_{B} \ket{z \oplus c(a+b)_n}_{Z}$.
        \For{$i = 1$ to $n - 1$} \Comment{Slice 1}
            \State $\CNOT(A_i,B_i)$ 
        \EndFor
        \State Toffoli$(C, A_{n-1}, Z)$ \Comment{Slice 2}
        \State Apply $\Lg_1^{(n-2)}$ on $(A_1, \dots, A_{n-1})$
        \State Apply $\Rbra{\Lg_2^{(n-1)}}^\dagger$ on $(A_0, B_0, \dots, A_{n-2}, B_{n-2}, A_{n-1})$ \Comment{Slice 3}
        \State $\MCX_3(C, A_{n-1}, B_{n-1}, Z)$
        \State Apply $\Fg_2^{(n-1)}$ on $(C, A_1, B_1, \dots, A_{n-1}, B_{n-1})$ \Comment{Slice 4}
        \State Apply $\Fg_1^{(n-2)}$ on $(C, B_1, \dots, B_{n-2})$ \Comment{Slice 5}
        \State Apply $\Lg_2^{(n-1)}$ on $(A_0, B_0, \dots, A_{n-2}, B_{n-2}, A_{n-1})$
        \State Apply $\Fg_1^{(n-2)}$ on $(C, B_1, \dots, B_{n-2})$
        \State Apply $\Rbra{\Lg_1^{(n-2)}}^\dagger$ on $(A_1, \dots, A_{n-1})$ \Comment{Slice 6}
        \State Toffoli$(C, A_{0}, B_{0})$ \Comment{Slice 7}
        \For{$i = 1$ to $n-1$}
            \State $\CNOT(A_i,B_i)$
        \EndFor
    \end{algorithmic}
\end{algorithm}

Algorithm~\ref{algo:controlled_adder} presents a pseudocode for generating a circuit implementing the controlled addition operator.
The structure of this algorithm is similar to that of Algorithm~\ref{algo:adder}, as indicated by the slices in the pseudocode.

The correctness and the depth and gate count complexities of the circuit produced by Algorithm~\ref{algo:controlled_adder} are established by Theorem~\ref{thm:controlled_adder}, which generalizes the result of Algorithm~\ref{algo:adder} to a controlled adder.
\begin{theorem}\label{thm:controlled_adder}
    Let $a$ and $b$ be two $n$-bit integers. The circuit produced by Algorithm~\ref{algo:controlled_adder} implements the in-place addition of $a$ and $b$ controlled by $c \in \{0,1\}$, $\ie$, an operator with the following action:
    \[
        \ket{c} \ket{a} \ket{b} \ket{z} \mapsto \ket{c} \ket{a} \ket{ca + b \mod 2^n} \ket{z \oplus c (a + b)_n}
    \]
     (where $z \in \{0,1\}$) and can be constructed over the $\{\text{Toffoli}, \CNOT\}$ gate set with a depth of $\bigO{\log^2 n}$ and a gate count of $\bigO{n \log n}$, without any ancilla qubits.
\end{theorem}
\begin{proof}
    We will rely on the following equality:
    \begin{equation}\label{eq:controlled_uvu}
        \begin{quantikz}[align equals at=1.5, column sep=.3cm, row sep={.6cm,between origins}]
            \lstick{$\ket{c}$} & & \ctrl{1} & \ctrl{1} & \ctrl{1} & \\
            & \qwbundle{m} & \gate{U} & \gate{V} & \gate{U^\dag} & \\
        \end{quantikz}
        =
        \begin{quantikz}[align equals at=1.5, column sep=.3cm, row sep={.6cm,between origins}]
            \lstick{$\ket{c}$} & & & \ctrl{1} & & \\
            & \qwbundle{m} & \gate{U} & \gate{V} & \gate{U^\dag} & \\
        \end{quantikz}
    \end{equation}
    \ifarXiv which can be easily proven by case distinction:
    \begin{itemize}
        \item In the case where $c = 1$, the unitaries $U, V$, and $U^\dag$ are applied in both circuits.
        \item In the case where $c = 0$, the left-hand side circuit is equal to the identity since no unitaries are applied. On the right-hand side, only $U$ and $U^\dag$ are applied, which cancel each other out because $UU^\dag = I$.
    \end{itemize} \fi
    
    An adder controlled by a qubit $C$ can be constructed from the circuit produced by Algorithm~\ref{algo:adder} simply by adding $C$ as a control to all the gates in the circuit. However, as a result of Equation~\ref{eq:controlled_uvu}, it is not necessary to control a significant number of operators that are initially computed and subsequently uncomputed in the adder.
    \begin{itemize}
        \item In Algorithm~\ref{algo:adder}, the $\CNOT$ circuit in Slice 7 is the dagger of the one in Slice 1 plus one additional $\CNOT$ gate acting on $(A_0,B_0)$. Thus, only this gate needs control by $C$, hence the Toffoli gate acting on $(C, A_0, B_0)$ in Slice 7 in Algorithm~\ref{algo:controlled_adder}.
        \item In Algorithm~\ref{algo:adder}, the $\Lg_1^{(n-1)}$ operator applied on $(A_1, \dots, A_{n-1}, Z)$ in Slice 2 can be implemented by first applying a $\CNOT$ gate on $(A_{n-1}, Z)$, and then applying $\Lg_1^{(n-2)}$ on $(A_1, \dots, A_{n-1})$. Adding $C$ as a control of the $\CNOT$ gate results in a Toffoli gate applied on $(C, A_{n-1}, Z)$, as done in Algorithm~\ref{algo:controlled_adder}. The $\Lg_1^{(n-2)}$ operator applied on $(A_1, \dots, A_{n-1})$ being the dagger of the $\Rbra{\Lg_1^{(n-2)}}^\dagger$ operator applied in Slice $6$, neither needs control by $C$.
        \item In Algorithm~\ref{algo:adder}, the $\Rbra{\Lg_2^{(n)}}^\dagger$ operator applied on $(A_0, B_0, \dots, A_{n-1}, B_{n-1}, Z)$ in Slice 3 can be implemented by first applying $\Rbra{\Lg_2^{(n-1)}}^\dagger$ on $(A_0, B_0, \dots, A_{n-2}, B_{n-2}, A_{n-1})$, and then applying a Toffoli gate on $(A_{n-1}, B_{n-1}, Z)$. Adding $C$ as a control of the Toffoli gate results in a $\MCX_3$ gate applied on $(C, A_{n-1}, B_{n-1}, Z)$, as done in Algorithm~\ref{algo:controlled_adder}. The $\Rbra{\Lg_2^{(n-1)}}^\dagger$ operator applied on $(A_0, B_0, \dots, A_{n-2}, B_{n-2}, A_{n-1})$ being the dagger of the $\Lg_2^{(n-1)}$ operator applied in Slice $5$, neither needs control by $C$. However, the $\Xg$ gates applied on the qubits $B_i$ in Slice 5 must be controlled by $C$, which corresponds to the $\Fg_1^{(n-2)}$ operators applied on $(C, B_1, \dots, B_{n-1})$ in Algorithm~\ref{algo:controlled_adder}.
        \item Finally, adding $C$ as a control of the gates applied in Slice $4$ of Algorithm~\ref{algo:adder} corresponds to applying $\Fg_2^{(n-1)}$ on $(C, A_0, B_0, \dots, A_{n-1}, B_{n-1})$ in Algorithm~\ref{algo:controlled_adder}.
    \end{itemize}
    Thus, the circuit produced by Algorithm~\ref{algo:controlled_adder} is equivalent to the circuit produced by Algorithm~\ref{algo:adder} controlled by a qubit $C$ and therefore implements the controlled addition operator.

    We now analyze the depth and gate complexities of implementing each slice of the circuit produced by Algorithm~\ref{algo:controlled_adder} over the $\{\text{Toffoli}, \CNOT\}$ gate set.
    The $\Lg_1^{(n-1)}$ and $\Rbra{\Lg_1^{(n-2)}}^\dag$ operators in Slices 2 and 6 can be implemented with a depth of $\bigO{\log n}$ and a $\CNOT$-count of $\bigO{n}$, as stated by Lemma~\ref{thm:ladder1}.
    The $\Rbra{\Lg_2^{(n)}}^\dagger$ and $\Lg_2^{(n-1)}$ operators in Slices 3 and 5 can be implemented with a depth of $\bigO{\log^2 n}$ and a $\CNOT$-count of $\bigO{n\log n}$, as stated by Lemma~\ref{thm:ladder2}.
    The $\Fg_1^{(n-2)}$ operators in Slices 5 can be implemented with a depth of $\bigO{\log n}$ and a $\CNOT$-count of $\bigO{n}$, as stated by Lemma~\ref{th:fanout}.
    The $\Fg_2^{(n-1)}$ operator in Slice 4 can be implemented with a depth of $\bigO{\log n}$ and a $\CNOT$-count of $\bigO{n}$, as stated by Lemma~\ref{lem:fanout_2}.
    The Toffoli gates in Slice 2 and 7, as well as the $\MCX_3$ gate in Slice 3 can be implemented in constant depth and with a constant number of gates.
    The two sub-circuits formed by the remaining $\CNOT$ gates in Slices 1 and 7 both have a constant depth equal to $1$, as all the $\CNOT$ gates are applied on different qubits.
    Thus, the circuit produced by Algorithm~\ref{algo:controlled_adder} can be implemented over the $\{\text{Toffoli}, \CNOT, \Xg\}$ gate set with a depth of $\bigO{\log^2 n}$ and a gate count of $\bigO{n\log n}$, and without any ancilla qubits.
\end{proof}
\section{Application to classical-quantum arithmetic}\label{sec:QCarith}

In this section, two corollaries of the previous results are presented. First, incrementation of a quantum register is considered. Second, we examine the more general case of the addition of a classical number to a quantum register.

\subsection{Quantum Incrementer} \label{sec:incrementer}

As with addition, it is possible to implement an approximate incrementer with no ancilla using the $\QFT$. 
However, to achieve an exact and efficient circuit, it seems preferable to use only classical reversible gates.
Nie $\etal$ proposed such an incrementer, which uses the technique of conditionally clean ancillae, has polylogarithmic depth, and uses a clean ancilla qubit.
Here we propose a circuit that also has polylogarithmic depth but uses a dirty ancillary qubit. 
See Table~\ref{table:Incrementors} for the complexities of the various known incrementers.

\begin{table}[ht]
    \hspace{1cm}
    \centering
    \begin{NiceTabular}{c||c|c|c|}[hvlines, first-col] \hline
        & Addition Technique                                        & Size     & Depth 	& Ancilla   \\ \hline\hline
        \Block[c]{1-1}{Fourier-\\[-.3ex]based} & $\QFT$ \cite{Dra02}        & $\bigO{n^2}$ & $\bigO{n}$  & 0       \\ \hline\hline
        \Block[c]{3-1}{Classical\\[-.3ex]reversible} & Gidney \cite{Gid15b}                                     & $\bigO{n}$  & $\bigO{n}$  & 1 dirty \\
        & Nie $\etal$ \cite{NZS24} & $\bigO{n}$  & $\bigO{\log^2 n}$  & 1 clean \\
        & Corollary~\ref{thm:incrementerDirty} & $\bigO{n \log n}$ & $\bigO{\log^2 n}$ & 1 dirty \\ \hline
        \CodeAfter
        \SubMatrix\{{3-1}{5-1}.[left-xshift=2mm]
    \end{NiceTabular}
    \caption{Asymptotic complexity of quantum incrementors. }\label{table:Incrementors}
\end{table}

As a corollary of Theorem~\ref{thm:adder} and Theorem~\ref{thm:controlled_adder}, we can prove the following result:

\begin{corollary}\label{thm:incrementerDirty}
    Let $v$ be a $n$-bit integer. There exists a circuit that implements the incrementation of $v$, $\ie$, an operator with the following action:
    \[
	\ket{v} \mapsto \ket{v + 1 \mod 2^n}
    \]
	over the $\{\text{Toffoli}, \CNOT, \Xg\}$ gate set with a depth of $\bigO{\log^2 n}$ and a gate count of $\bigO{n \log n}$, with only one dirty ancilla qubit.
\end{corollary}
\begin{proof}
    As mentioned in \cite{HRS17}, it is possible to increment a $n$-qubit register $\ket{v}$ using another $n$-qubit register of dirty ancilla qubits that we denote by $\ket{g}$, as follows:
    \begin{align}
        \ket{v} \ket{g} & \longmapsto \ket{v-g} \ket{g} \label{eq:increment1} \\
        & \longmapsto \ket{v-g} \ket{\bar{g}} \label{eq:increment2} \\
        & \longmapsto \ket{v-g-\bar{g}} \ket{\bar{g}} = \ket{v+1} \ket{\bar{g}} \label{eq:increment3} \\
        & \longmapsto \ket{v+1} \ket{g} \label{eq:increment4}
    \end{align}
    where $\bar{g}$ denotes the bitwise complement of $g$.
    The mappings in Eq.~\ref{eq:increment2} and Eq.~\ref{eq:increment4} are done thanks to walls of $\Xg$ gates. 
    The mappings in Eq.~\ref{eq:increment1} and Eq.~\ref{eq:increment3} are done using a circuit for subtraction, $\ie$, a reversed circuit for addition. 
    Using Theorem~\ref{thm:adder}, we can thus state that a $n$-qubit incrementer can be implemented with a circuit over the $\{\text{Toffoli}, \CNOT, \Xg\}$ gate set with a depth of $\bigO{\log^2 n}$ and a gate count of $\bigO{n \log n}$, with $n$ dirty ancilla qubits. 
    As for a controlled version of this incrementer, using Theorem~\ref{thm:controlled_adder} we can prove the same result.
    
    Then, plugging it into the following equality, due to Gidney \cite{Gid15b}:
    \begin{equation*}\label{eq:SplitBorrowed}
        \begin{quantikz}[align equals at=2, row sep={.7cm,between origins}]
            \lstick[2]{$\ket{v}$} & \qwbundle{\ceil{\frac{n}{2}}} &            &          & \ctrl{2}&           & \ctrl{2}  &          & \gate{+1} & \rstick[2]{$\ket{v+1 \mod 2^n}$} \\
            & \qwbundle{\floor{\frac{n}{2}}} & \gate{+1}  & \targ{}  &         & \gate{+1} &           & \targ{}  && \\
            \lstick{$\ket{g}$} && \ctrl{-1}  & \ctrl{-1}& \targ{} & \ctrl{-1} & \targ{}   & \ctrl{-1}&& \rstick{$\ket{g}$}
        \end{quantikz}
    \end{equation*}
    where $\ket{g}$ is a dirty ancilla qubit, we can prove the claimed result. Indeed, with this circuit, a $n$-bit incrementer with 1 dirty ancillary qubit is transformed into three $n/2$-bit incrementers (among which two are a controlled version) with $n/2$ dirty ancillary qubits. As mentioned above, these can be implemented with a depth of $\bigO{\log^2 n}$ and a gate count of $\bigO{n \log n}$ over the $\{\text{Toffoli}, \CNOT, \Xg\}$ gate set. Apart from these three incrementers, two $\Fg_1^{(\lfloor n/2 \rfloor)}$ gates that can be implemented with a $\CNOT$-depth of $\bigO{\log n}$ and a $\CNOT$-count of $\bigO{n}$ (Lemma~\ref{th:fanout}) and two $\MCX_{\lceil \frac{n}{2} \rceil}$ gates are used in the circuit, with the latter being implementable in logarithmic depth and a linear number of gates (Theorem~\ref{th:MCX}). Hence the result.
\end{proof}

\subsection{Quantum Constant Adder} \label{sec:cst_adder}

Finally, we move on to the addition of a classically known number in a quantum register. A naive method of implementing this operation is to allocate a quantum register and write the constant into it, then use a quantum adder such as the one described in Section~\ref{sec:adder} to perform the desired addition. This method obviously requires a large number of clean ancilla qubits. 
Better still, a method using only one dirty ancilla qubit exists \cite{HRS17}. We use it to propose the first classical-quantum adder without ancilla, with sublinear depth. 
See Table~\ref{table:ClAdders} for the complexities of the various known classical-quantum adders.

\begin{table}[ht]
    \hspace{1cm}
    \centering
    \begin{NiceTabular}{c||c|c|c|}[hvlines, first-col] \hline
        & Addition Technique                                        & Size     & Depth 	& Ancilla   \\ \hline\hline
        \Block[c]{1-1}{Fourier-\\[-.3ex]based} & $\QFT$ \cite{Dra02}        & $\bigO{n^2}$ & $\bigO{n}$  & 0       \\ \hline\hline
        \Block[c]{3-1}{Classical\\[-.3ex]reversible} & Takahashi $\etal$ \cite{TTK10}                                     & $\bigO{n}$  & $\bigO{n}$  & $n$ clean \\
        & H{\"a}ner $\etal$ \cite{HRS17} & $\bigO{n \log n}$ & $\bigO{n}$ & 1 dirty \\
        & Corollary~\ref{cor:addconst} & $\bigO{n \log^2 n}$ & $\bigO{\log^3 n}$ & 1 dirty \\ \hline
        \CodeAfter
        \SubMatrix\{{3-1}{5-1}.[left-xshift=2mm]
    \end{NiceTabular}
    \caption{Asymptotic complexity of classical-quantum adders. }\label{table:ClAdders}
\end{table}

As a corollary of Theorem~\ref{thm:adder} and Theorem~\ref{thm:controlled_adder}, we can prove the following result:

\begin{corollary}\label{cor:addconst}
    Let $v$ and $c$ be two $n$-bit integers. There exists a circuit that implements the in-place addition of $c$ in a quantum register storing $v$, $\ie$, an operator with the following action:
    \[
	\ket{v} \mapsto \ket{v + c \mod 2^n}
    \]
	over the $\{\text{Toffoli}, \CNOT, \Xg\}$ gate set with a depth of $\bigO{\log^3 n}$ and a gate count of $\bigO{n \log^2 n}$, without any ancilla qubits.
\end{corollary}
\begin{proof}
    We start from the constant adder with one dirty ancilla proposed by H{\"a}ner $\etal$ \cite{HRS17}:
    \begin{equation*}\label{eq:AddConst}
        \begin{quantikz}[align equals at=2, row sep={.7cm,between origins}]
            \lstick[2]{$\ket{v}$} & \qwbundle{\ceil{\frac{n}{2}}} &            &          & \gate[label style={black,rotate=90}, wires=3]{CARRY}&          & \gate[label style={black,rotate=90}, wires=3]{CARRY}&          & \gate{+c_\ell}& \rstick[2]{$\ket{v+c \mod 2^n}$} \\
            & \qwbundle{\floor{\frac{n}{2}}}    & \gate{+1}  & \targ{}  &               & \gate{+1} &                & \targ{}  & \gate{+c_h}   & \\
            \lstick{$\ket{g}$} &                   & \ctrl{-1}  & \ctrl{-1} &              & \ctrl{-1} &                & \ctrl{-1}&               & \rstick{$\ket{g}$}
        \end{quantikz}
    \end{equation*}
    where $\ket{g}$ is a dirty ancilla qubit. The CARRY operator acts as follows:
    \begin{equation*}\label{eq:CARRY}
        \begin{quantikz}[align equals at=2, row sep={.7cm,between origins}]
            \lstick{$\ket{v_\ell}$} & \qwbundle{\ceil{\frac{n}{2}}}   & \gate[label style={black,rotate=90}, wires=3]{CARRY}& \rstick{$\ket{v_\ell}$} \\
            \lstick{$\ket{v_h}$} & \qwbundle{\floor{\frac{n}{2}}-1} & & \rstick{$\ket{v_h}$} \\
            \lstick{$\ket{g}$} &                & & \rstick{$\ket{g \oplus (v_\ell+c_\ell)_{\ceil{\frac{n}{2}}}}$}
        \end{quantikz}
    \end{equation*}
    where $\ket{v_\ell}$ is a register storing the least significant bits of $v$ while $\ket{v_h}$ is a register storing the most significant bits of $v$. The latter is used as a register of dirty ancilla qubits to compute the carry bit of the addition of $v_\ell$ with the corresponding least significant bits of the constant $c$, that we denote $c_\ell$. This circuit allows breaking down the addition of a $n$-bit constant into two additions of $n/2$-bit constants acting on different qubits. 
    
    By Corollary~\ref{thm:incrementerDirty}, the controlled incrementers of the circuit can be implemented with a Toffoli-depth of $\bigO{\log^2 n}$, a Toffoli-count of $\bigO{n \log n}$ and $n/2$ dirty ancillae, that are indeed available.
    By Lemma~\ref{th:fanout}, the fan-out operators have $\CNOT$-depth $\bigO{\log n}$ and $\CNOT$-count $\bigO{n}$.
    Then, the CARRY operators have Toffoli-depth $\bigO{\log^2 n}$, Toffoli-count $\bigO{n \log n}$ and use $\floor{n/2}$ dirty ancillae, thanks to Lemma~\ref{thm:ladder1} and Lemma~\ref{thm:ladder2} (we do not go into the details here, the implementation of CARRY is given in \cite{HRS17} and we basically replace the $\CNOT$ and Toffoli ladders with our implementations). 
    Finally, the recursive calls are executed sequentially for the first level of the recursion (so that we have enough dirty ancilla qubits for both operators) and then in parallel for the recursive calls.
    
    Thus, the depth $D(n)$ of the circuit is computed thanks to:
    \begin{equation}
        D(i) = 
        \begin{cases}
            2 D(i/2) + O(\log^2 i) & \text{ for the first level (i.e. when $i=n$) } \\
            D(i/2) + O(\log^2 i) & \text{ for the recursive calls (i.e. when $i=n/2^j$ for some $j>1$). } \\
        \end{cases}
    \end{equation}
    Hence:
    \[
        D(n) = O(\log^3 n) .
    \]
    
    For the gate count $C(n)$ of the circuit, we have:
    \begin{equation}
        C(n) = 2 C(n/2) + O(n \log n) = O(n \log^2 n).
    \end{equation}
\end{proof}
\section{Discussion}

We have proposed a novel quantum (in fact, reversible classical) adder implementation based on the ripple-carry technique and without ancilla qubits. This ripple-carry adder, unlike its predecessors, which have a linear depth for a linear number of gates (over the $\{\text{Toffoli}, \CNOT, \Xg\}$ gate set), exhibits a polylogarithmic depth for a linearithmic number of gates (over the same gate set). This results in an exponential reduction in the depth of quantum ripple-carry adders. Our work demonstrates the existence of a quantum adder based on reversible classical logic that offers a promising alternative to the prominent $\QFT$-based adder (exhibiting inherent limitations in the use of small-angle rotation gates). 
Furthermore, we have shown that the controlled version of this adder retains the same properties.

Our results corroborate the findings of Nie $\etal$ \cite{NZS24}, which introduced a quantum incrementer circuit with the same properties of polylogarithmic depth and classical reversible logic only. Furthermore, this new algorithm for addition also lends support to Khattar and Gidney's statement \cite{KG24}: the use of conditionally clean ancillae definitely seems to be an essential technique in the design of more efficient quantum algorithms.

Finally, our results are based on novel low-depth implementations of the $\CNOT$ ladder and Toffoli ladder operators.
The simplicity of these operators suggests that they are likely to appear in various quantum circuits.
As such, our constructions have the potential to lead to significant depth reduction in other quantum circuits.
For example, it has been shown in~\cite{Van25} that our logarithmic-depth implementation of the $\CNOT$ ladder operator enables binary field multiplication to be performed in logarithmic depth for certain primitive polynomials, such as trinomials or equally spaced polynomials.


\subsection*{Acknowledgments}

This work is part of HQI initiative (www.hqi.fr) and is supported by France 2030 under the French National Research Agency award number “ANR-22-PNCQ-0002”.

\bibliographystyle{quantum}
\bibliography{biblio}

\appendix
\section{Proof of Lemma~\ref{thm:ladder_alpha}}\label{sec:proof_ladder_alpha}

Let us recall Lemma~\ref{thm:ladder_alpha}: \thmladderalpha* 

\begin{proof}
    We first prove that the circuit produced by Algorithm~\ref{algo:ladder_alpha} implements $\Lg_{\bs \alpha}$.
    For the base case where $k = 1$ (meaning that $\bs \alpha$ is empty), $\Lg_{\bs \alpha}$ is equal to the identity operator, which corresponds to the empty circuit returned by Algorithm~\ref{algo:ladder_alpha}.
    For the other base case where $k = 2$, the algorithm produces a circuit containing a single $\MCX$ gate applied on the qubits $(X_0, \ldots, X_{\alpha_0})$, which corresponds to the implementation of the $\Lg_{\bs \alpha}$ operator:
    \begin{equation}
        \MCX\left(\bigotimes_{i=0}^{\alpha_0} \ket{x_i}\right) = \left(\bigotimes_{i=0}^{\alpha_0-1} \ket{x_i}\right) \lvert x_{\alpha_0} \oplus \prod_{i=0}^{\alpha_0-1} x_i\rangle
             = \Lg_{\bs \alpha} \left(\bigotimes_{i=0}^{\alpha_0} \ket{x_i}\right).
    \end{equation}
    For the other cases where $k > 2$, Algorithm~\ref{algo:ladder_alpha} constructs two circuits, $C_L$ and $C_R$, and performs a recursive call with parameters $\bs \alpha'$ and a subset of qubits $X'$, which produces a circuit that we denote by $C_{X'}$.
    The circuit produced by Algorithm~\ref{algo:ladder_alpha} is then the result of the concatenation of the circuits $C_L$, $C_{X'}$, and $C_R$.
    Let $\Ug_L$, $\Ug_{X'}$, and $\Ug_R$ be the unitary operators associated with the circuits $C_L$, $C_{X'}$, and $C_R$, respectively. Let $\ket{\bs x}$ be an $(\alpha_{k-2}+1)$-dimensional computational basis state.
    
    The $\Ug_L$ operator acts on $\ket{\bs x}$ as follows:
    \begin{equation*}
        \begin{aligned}
            \Ug_L\ket{\bs x} = &\bigotimes_{i=0}^{\alpha_0} \ket{x_i} \left(\bigotimes_{i=1}^{\ceil{\frac{k}{2}}-2} \left(\bigotimes_{j=\alpha_{2i-2}+1}^{\alpha_{2i-1}-1} \ket{x_j}\right) \lvert x_{\alpha_{2i-1}} \oplus \prod_{j=\alpha_{2i-2}}^{\alpha_{2i-1}-1} x_j\rangle \bigotimes_{j=\alpha_{2i-1}+1}^{\alpha_{2i}} \ket{x_j}\right) \\
                             &\left(\bigotimes_{k\bmod{2}}^{0} \bigotimes_{i=\alpha_{k-4}+1}^{\alpha_{k-3}} \ket{x_{i}}\right) \left(\bigotimes_{i=\alpha_{k-3}+1}^{\alpha_{k-2}-1} \ket{x_{i}}\right) \lvert x_{\alpha_{k-2}}\oplus \prod_{j=\alpha_{k-3}}^{\alpha_{k-2}-1} x_j\rangle.
        \end{aligned}
    \end{equation*}
    The operator $\Ug_{X'}$ implements the $\Lg_{\bs \alpha'}$ operator on the $\alpha'_{\lfloor k/2 \rfloor-2}+1$ qubits $X'$. Thus, it acts on $\ket{\bs x}$ as follows:
    \begin{equation*}
        \begin{aligned}
        \Ug_{X'}\ket{\bs x} = &\bigotimes_{i=0}^{\alpha_0} \ket{x_i} \left(\bigotimes_{i=1}^{\ceil{\frac{k}{2}}-2} \left(\bigotimes_{j=\alpha_{2i-2}+1}^{\alpha_{2i-1}} \ket{x_j}\right) \left(\bigotimes_{j=\alpha_{2i-1}+1}^{\alpha_{2i}-1} \ket{x_j}\right) \lvert x_{\alpha_{2i}} \oplus \prod_{\substack{j=\alpha_{2i-2} \\ j \neq \alpha_{2i-1}}}^{\alpha_{2i}-1} x_j \rangle\right) \\
                                  &\left(\bigotimes_{k\bmod{2}}^{0} \left(\bigotimes_{i=\alpha_{k-4}+1}^{\alpha_{k-3}-1} \ket{x_{i}}\right)\lvert x_{\alpha_{k-3}} \oplus \prod_{j=\alpha_{k-4}}^{\alpha_{k-3}-1} x_j \rangle\right) \bigotimes_{i=\alpha_{k-3}+1}^{\alpha_{k-2}} \ket{x_{i}}.
        \end{aligned}
    \end{equation*}
    The $\Ug_R$ operator acts on $\ket{\bs x}$ as follows:
    \begin{equation*}
        \begin{aligned}
            \Ug_R\ket{\bs x} = &\left(\bigotimes_{i=0}^{\alpha_0-1} \ket{x_i}\lvert x_{\alpha_0}\oplus \prod_{i=0}^{\alpha_0-1} x_i\rangle\right) \!\left(\!\bigotimes_{i=1}^{\ceil{\frac{k}{2}}-2} \!\left(\bigotimes_{j=\alpha_{2i-2}+1}^{\alpha_{2i-1}} \!\!\!\ket{x_j}\right) \!\left(\bigotimes_{j=\alpha_{2i-1}+1}^{\alpha_{2i}-1} \!\!\!\ket{x_j}\right) \lvert x_{\alpha_{2i}} \oplus \!\!\prod_{j=\alpha_{2i-1}}^{\alpha_{2i}-1} x_j\rangle\right) \\
                           &\left(\bigotimes_{k\bmod{2}}^{0} \bigotimes_{i=\alpha_{k-4}+1}^{\alpha_{k-3}} \ket{x_{i}}\right) \bigotimes_{i=\alpha_{k-3}+1}^{\alpha_{k-2}} \ket{x_{i}}.
        \end{aligned}
    \end{equation*}
    By putting these equations together, the $\Ug_R \Ug_{X'} \Ug_L$ operator acts on $\ket{\bs x}$ as follows:
    \begin{equation*}
        \begin{aligned}
            &\Ug_R\Ug_{X'}\Ug_L\ket{\bs x} \\
            &= \Ug_R\Ug_{X'}\left[\bigotimes_{i=0}^{\alpha_0} \ket{x_i} \left(\bigotimes_{i=1}^{\ceil{\frac{k}{2}}-2} \left(\bigotimes_{j=\alpha_{2i-2}+1}^{\alpha_{2i-1}-1} \ket{x_j}\right) \lvert x_{\alpha_{2i-1}} \oplus \prod_{j=\alpha_{2i-2}}^{\alpha_{2i-1}-1} x_j\rangle \bigotimes_{j=\alpha_{2i-1}+1}^{\alpha_{2i}} \ket{x_j}\right) \right.\\
            &\quad \left. \left(\bigotimes_{k\bmod{2}}^{0} \bigotimes_{i=\alpha_{k-4}+1}^{\alpha_{k-3}} \ket{x_{i}}\right) \left(\bigotimes_{i=\alpha_{k-3}+1}^{\alpha_{k-2}-1} \ket{x_{i}}\right) \lvert x_{\alpha_{k-2}}\oplus \prod_{j=\alpha_{k-3}}^{\alpha_{k-2}-1} x_j\rangle \right] \\
            &= \Ug_R\left[\bigotimes_{i=0}^{\alpha_0} \ket{x_i} \left(\bigotimes_{i=1}^{\ceil{\frac{k}{2}}-2} \left(\bigotimes_{j=\alpha_{2i-2}+1}^{\alpha_{2i-1}-1} \ket{x_j}\right) \lvert x_{\alpha_{2i-1}} \oplus \prod_{j=\alpha_{2i-2}}^{\alpha_{2i-1}-1} x_j\rangle \right. \right. \\
            &\quad \left. \left(\bigotimes_{j=\alpha_{2i-1}+1}^{\alpha_{2i}-1} \ket{x_j}\right) \lvert x_{\alpha_{2i}} \oplus \prod_{\substack{j=\alpha_{2i-2} \\ j \neq \alpha_{2i-1}}}^{\alpha_{2i}-1} x_j \rangle\right) \\
            &\quad \left. \left(\bigotimes_{k\bmod{2}}^{0} \left(\bigotimes_{i=\alpha_{k-4}+1}^{\alpha_{k-3}-1} \ket{x_{i}}\right)\lvert x_{\alpha_{k-3}} \oplus \prod_{j=\alpha_{k-4}}^{\alpha_{k-3}-1} x_j \rangle\right) \left(\bigotimes_{i=\alpha_{k-3}+1}^{\alpha_{k-2}-1} \ket{x_{i}}\right) \lvert x_{\alpha_{k-2}}\oplus \prod_{j=\alpha_{k-3}}^{\alpha_{k-2}-1} x_j\rangle \right] \\
            &= \left(\bigotimes_{i=0}^{\alpha_0-1} \ket{x_i}\lvert x_{\alpha_0}\oplus \prod_{j=0}^{\alpha_0-1} x_j\rangle\right) \left(\bigotimes_{i=1}^{\ceil{\frac{k}{2}}-2} \left(\bigotimes_{j=\alpha_{2i-2}+1}^{\alpha_{2i-1}-1} \ket{x_j}\right) \lvert x_{\alpha_{2i-1}} \oplus \prod_{j=\alpha_{2i-2}}^{\alpha_{2i-1}-1} x_j\rangle \right. \\
            &\quad \left. \left(\bigotimes_{j=\alpha_{2i-1}+1}^{\alpha_{2i}-1} \ket{x_j}\right) \lvert x_{\alpha_{2i}} \oplus \prod_{\substack{j=\alpha_{2i-2} \\j \neq \alpha_{2i-1}}}^{\alpha_{2i}-1} x_j \oplus \left(x_{\alpha_{2i-1}} \oplus \prod_{j=\alpha_{2i-2}}^{\alpha_{2i-1}-1} x_j\right) \prod_{\alpha_{2i-1}+1}^{\alpha_{2i}-1} x_j\rangle\right) \\
            &\quad \left(\bigotimes_{k\bmod{2}}^{0} \left(\bigotimes_{i=\alpha_{k-4}+1}^{\alpha_{k-3}-1} \ket{x_{i}}\right)\lvert x_{\alpha_{k-3}} \oplus \prod_{j=\alpha_{k-4}}^{\alpha_{k-3}-1} x_j \rangle\right) \left(\bigotimes_{i=\alpha_{k-3}+1}^{\alpha_{k-2}-1} \ket{x_{i}}\right) \lvert x_{\alpha_{k-2}}\oplus \prod_{j=\alpha_{k-3}}^{\alpha_{k-2}-1} x_j\rangle \\
            &= \left(\bigotimes_{i=0}^{\alpha_0-1} \ket{x_i}\lvert x_{\alpha_0}\oplus \prod_{j=0}^{\alpha_0-1} x_j\rangle\right) \left(\bigotimes_{i=1}^{\ceil{\frac{k}{2}}-2} \left(\bigotimes_{j=\alpha_{2i-2}+1}^{\alpha_{2i-1}-1} \ket{x_j}\right) \lvert x_{\alpha_{2i-1}} \oplus \prod_{j=\alpha_{2i-2}}^{\alpha_{2i-1}-1} x_j\rangle \right. \\
            &\quad \left. \left(\bigotimes_{j=\alpha_{2i-1}+1}^{\alpha_{2i}-1} \ket{x_j}\right) \lvert x_{\alpha_{2i}} \oplus \prod_{\alpha_{2i-1}}^{\alpha_{2i}-1} x_j\rangle\right) \\
            &\quad \left(\bigotimes_{k\bmod{2}}^{0} \left(\bigotimes_{i=\alpha_{k-4}+1}^{\alpha_{k-3}-1} \ket{x_{i}}\right)\lvert x_{\alpha_{k-3}} \oplus \prod_{j=\alpha_{k-4}}^{\alpha_{k-3}-1} x_j \rangle\right) \left(\bigotimes_{i=\alpha_{k-3}+1}^{\alpha_{k-2}-1} \ket{x_{i}}\right) \lvert x_{\alpha_{k-2}}\oplus \prod_{j=\alpha_{k-3}}^{\alpha_{k-2}-1} x_j\rangle \\
           &= \left(\bigotimes_{i=0}^{\alpha_0-1} \ket{x_i} \lvert x_{\alpha_0} \oplus \prod_{j=0}^{\alpha_0 - 1}x_j \rangle\right) \left(\bigotimes_{i=1}^{k-2}\left(\bigotimes_{j=\alpha_{i-1}+1}^{\alpha_i-1}\ket{x_j}\right)\lvert x_{\alpha_i}\oplus \prod_{j=\alpha_{i-1}}^{\alpha_i-1} x_j \rangle\right) \\
           &= \Lg_{\bs \alpha}\ket{\bs x}
        \end{aligned}
    \end{equation*}
    
    Thus, the circuit produced by Algorithm~\ref{algo:ladder_alpha}, associated with the operator $\Ug_R\Ug_{X'}\Ug_L$, produces a circuit implementing the $\Lg_{\bs \alpha}$ operator.

    The $\MCX$-depth of the $C_L$ and $C_R$ circuits is exactly one, because all the $\MCX$ gates in these circuits are applied on different qubits.
    Therefore, the depth $D(k)$ of the circuit produced by Algorithm~\ref{algo:ladder_alpha} is
    \begin{equation}
        D(k) = 2 + D\left(\floor{\frac{k}{2}}\right)
    \end{equation}
    with $D(2) = 1$ and $D(3) = 2$.
    This equation is equivalent to Equation~\ref{eq:depth_ladder1}, which was demonstrated in the proof of Theorem~\ref{thm:ladder1} to be equal, for $k \geq 2$, to 
    \begin{equation}
        \begin{aligned}
            D(k) &= \floor{\log(k)} + \floor{\log\left(\frac{2k}{3}\right)}.
        \end{aligned}
    \end{equation}

    Finally, the number of $\MCX$ gates in the $C_L$ and $C_R$ circuits is
    \begin{equation}
        \floor{\frac{k-1}{2}},
    \end{equation}
    which implies that the number of $\MCX$ gates in the circuit produced by Algorithm~\ref{algo:ladder_alpha} is
    \begin{equation}
        C(k) = 2\floor{\frac{k-1}{2}} + C\left(\floor{\frac{k}{2}}\right)
    \end{equation}
    with $C(2) = 1$ and $C(3)=2$.
    This equation is equivalent to Equation~\ref{eq:cnot_ladder1}, which was demonstrated in the proof of Theorem~\ref{thm:ladder1} to be equal, for $k \geq 2$, to 
    \begin{equation}
        \begin{aligned}
            C(k) &= 2k - 2 - \floor{\log(k)} - \floor{\log\left(\frac{2k}{3}\right)}.
        \end{aligned}
    \end{equation}
\end{proof}

\end{document}